\documentclass[fleqn,10pt]{article}
\usepackage[a4paper,left=3cm,right=2cm]{geometry}
\usepackage[utf8]{inputenc}
\usepackage[T1]{fontenc}
\usepackage[symbol]{footmisc}
\usepackage{authblk}
\usepackage{siunitx}
\usepackage{booktabs}
\usepackage{amsmath,amssymb,amsthm}
\usepackage{cases}
\usepackage{url}
\usepackage{bm}
\usepackage{graphicx}
\usepackage{enumitem}
\usepackage{algorithm}
\usepackage{float}
\usepackage{mathtools}
\usepackage{todonotes}

\setlist[description]{leftmargin=*}

\theoremstyle{plain}
\newtheorem{thm}{Theorem}[section]
\newtheorem{lem}{Lemma}[section]
\newtheorem{cor}{Corollary}[section]

\theoremstyle{remark}
\newtheorem{rem}{Remark}[section]

\newcommand{\set}[1]{\left\{#1\right\}}
\newcommand{\RR}{\mathbb{R}}
\newcommand{\NN}{\mathbb{N}}
\newcommand{\Lor}{\mathbb{R}^{1,d}}
\newcommand{\Rplus}{\mathbb{R}_{\geqslant 0}}

\newcommand{\ee}{\bm{e}}
\newcommand{\oo}{\bm{o}}

\newcommand{\qq}{\bm{q}}

\newcommand{\xx}{\bm{x}}
\newcommand{\yy}{\bm{y}}

\newcommand{\dist}{\mathrm{d}}

\newcommand{\cO}{\mathcal{O}}
\newcommand{\cH}{\mathcal{H}}
\newcommand{\norm}[1]{\left\Vert#1\right\Vert}
\newcommand{\diag}{\operatorname{diag}}
\newcommand{\argmin}{\operatorname{arg\,min}}
\newcommand{\strain}{\operatorname{Strain}}
\newcommand{\stress}{\operatorname{Stress}}

\DeclareMathOperator{\arcosh}{arcosh}

\makeatletter
\newcommand\ackname{Acknowledgements}
\if@titlepage
  \newenvironment{acknowledgements}{%
      \titlepage
      \null\vfil
      \@beginparpenalty\@lowpenalty
      \begin{center}%
        \bfseries \ackname
        \@endparpenalty\@M
      \end{center}}%
     {\par\vfil\null\endtitlepage}
\else
  \newenvironment{acknowledgements}{%
      \if@twocolumn
        \section*{\abstractname}%
      \else
        \small
        \begin{center}%
          {\bfseries \ackname\vspace{-.5em}\vspace{\z@}}%
        \end{center}%
        \quotation
      \fi}
      {\if@twocolumn\else\endquotation\fi}
\fi
\makeatother

\begin{document}
\pagestyle{empty}

\title{Strain-Minimizing Hyperbolic Network Embeddings with Landmarks}
\author[1,*]{Martin Keller-Ressel}
\author[1,$\dagger$]{Stephanie Nargang}
\affil[1]{TU Dresden, Institute for Mathematical Stochastics, Dresden, 01062, Germany}
\affil[*]{martin.keller-ressel@tu-dresden.de}
\affil[$\dagger$]{stephanie.nargang@tu-dresden.de}

\maketitle

\begin{abstract}
{We introduce L-hydra (landmarked hyperbolic distance recovery and approximation), a method for embedding network- or distance-based data into hyperbolic space, which requires only the distance measurements to a few `landmark nodes'. This landmark heuristic makes L-hydra applicable to large-scale graphs and improves upon previously introduced methods. As a mathematical justification, we show that a point configuration in $d$-dimensional hyperbolic space can be perfectly recovered (up to isometry) from distance measurements to just $d+1$ landmarks. We also show that L-hydra solves a two-stage strain-minimization problem, similar to our previous (unlandmarked) method `hydra'. Testing on real network data, we show that L-hydra is an order of magnitude faster than existing hyperbolic embedding methods and scales linearly in the number of nodes. While the embedding error of L-hydra is higher than the error of existing methods, we introduce an extension, L-hydra+, which outperforms existing methods in both runtime and embedding quality.}
\end{abstract}

\begin{acknowledgements} We would like to thank Kenny Chowdhary and Tamara Kolda for providing us with the source code of the \texttt{HyPy} algorithm from \cite{chowdhary2017improved}.
\end{acknowledgements}

\section{Introduction}
Embeddings of networks and distance-based data into hyperbolic geometry have received substantial attention in the recent decade. Such embeddings have been used for link prediction \cite{papadopoulos2012popularity, papadopoulos2015network}, visualization \cite{walter2004h}, and community detection \cite{papadopoulos2015network, muscoloni2017machine} in networks. In addition to providing insight into the tradeoff between popularity and similarity effects in network growth \cite{papadopoulos2012popularity}, they have interesting implications for routing, network navigability \cite{kleinberg2007geographic, boguna2009navigability} and efficient computation of shortest paths \cite{zhao2011fast, chowdhary2017improved}. In \cite{keller2019hydra} we have introduced \texttt{hydra} (hyperbolic distance recovery and approximation) and \texttt{hydra+} as efficient methods to compute such embeddings. Similar to classic multidimensional scaling in Euclidean space, \texttt{hydra} uses an Eigendecomposition technique to minimize \emph{strain}, a functional based on the hyperbolic Gram matrix of embedded points. Due to this technique, \texttt{hydra} provides an efficient alternative to methods like \texttt{Rigel} of \cite{zhao2011fast} and \texttt{HyPy} of \cite{chowdhary2017improved}, which minimize \emph{stress}, i.e., the least-squares embedding error, by numerical optimization. However, as \texttt{hydra} requires the full shortest path matrix of the input network, it can not be scaled to large and very large networks beyond the order of 100.000 nodes. \texttt{Rigel} and \texttt{HyPy} on the other side use a clever landmark heuristic (see also \cite{de2004sparse}), which replaces the full shortest path matrix by only the shortest paths terminating at a small number of (randomly selected) landmark nodes. Here, we show that this landmark heuristic can also be applied to the method of strain minimization, and introduce \texttt{L-hydra} (Landmarked hyperbolic distance recovery and approximation) as a new hyperbolic embedding method for very large networks. On the theoretical side, we show that \texttt{L-hydra} perfectly recovers (up to isometry) any $n$-point configuration in $d$-dimensional hyperbolic space from its distances measured to just $d+1$ landmarks. This result provides, for the first time, a theoretical basis to the use of landmark methods for hyperbolic network embeddings. In particular it shows that the required number of landmarks depends only on the intrinsic dimension $d$ of the network, but not on the number $n$ of vertices; hence the computational load of \texttt{L-hydra} scales only \emph{linearly} with the number $n$ of vertices. Finally, we present in Section~\ref{sec:numerics} numerical results for \texttt{L-hydra} and its extension \texttt{L-hydra+}, showing that they provide an improvement over \texttt{HyPy} in both runtime and embedding error.

\section{Background}

\subsection{Hyperbolic space}
Our embedding method is formulated in the mathematical framework of the $d$-dimensional hyperboloid model of hyperbolic geometry (cf. \cite{ratcliffe2006foundations, cannon1997hyperbolic}). For $\xx, \yy \in \RR^{d+1}$ the \emph{Lorentz product}, an indefinite inner product, is defined by
\begin{equation}\label{eq:lorentz}
\xx \circ \yy := x_1 y_1 - \left(x_2 y_2 + \dotsc + x_{d+1} y_{d+1}\right).
\end{equation}
The real vector space $\RR^{d+1}$ equipped with this inner product is called \emph{Lorentz space} and denoted by $\Lor$. It contains, as subset, the \emph{positive Lorentz space} $\Lor_+ = \set{\xx \in \Lor: x_1 > 0}$.  Within $\Lor_+$, the single-sheet hyperboloid $\cH_d$ is given by 
\begin{equation}\label{eq:Hd}
\cH_d = \set{\xx \in \Lor: \xx \circ \xx = 1, x_1 > 0}.
\end{equation}
The \emph{hyperboloid model} in dimension $d$ with curvature $-\kappa$, $(\kappa >0)$, consists of $\cH_d$ endowed with the hyperbolic distance
\begin{equation}\label{eq:hyper_dist}
\dist^\kappa_H(\xx,\yy) = \frac{1}{\sqrt{\kappa}}\arcosh\left(\xx \circ \yy\right), \qquad \xx, \yy \in \cH_d.
\end{equation}
The hyperbolic distance $\dist^\kappa_H$ is a distance on $\cH_d$ in the usual mathematical sense; in particular it takes only positive values and satisfies the triangle inequality, cf. \cite[\S3.2]{ratcliffe2006foundations}. In fact, equipped with the metric tensor $ds^2 = \frac{1}{\kappa}(d\xx \circ d\xx)$, the hyperboloid $\cH_d$ becomes a Riemannian manifold of constant sectional curvature $-\kappa$ and $\dist^\kappa_H$ is exactly its geodesic distance. Note that the curvature parameter does not enter into the description of the hyperboloid $\cH_d$, but only in the distance metric $\dist^\kappa_H$. Just as Euclidean space is the canonical model of geometry with zero curvature, hyperbolic space is the canonical model of geometry with negative curvature. Aside from the hyperbolid model, other equivalent models of hyperbolic geometry exist, including the Poincar\'e ball model and the upper half-space model, see \cite[\S4.2, \S4.6]{ratcliffe2006foundations} and \cite{cannon1997hyperbolic}.

\subsection{Embedding of distances and networks}
To formulate the problem of embedding network or other data into hyperbolic space, let some objects $\oo_1, \dotsc, \oo_n$ be given. Let $D = [d_{ij}] \in \Rplus^{n \times n}$ be a symmetric matrix with zero diagonal, which represents the induced pairwise dissimilarities between the objects. The goal is to find a low dimensional coordinate representation $\xx_1, \dotsc, \xx_n$ of the objects in hyperbolic space $\cH_d$, such that the hyperbolic distances between the coordinate representations approximate the given dissimilarities as closely as possible, i.e., such that
\begin{equation}\label{eq:embed}
\dist^\kappa_H(\xx_i, \xx_j) \approx d_{ij}.
\end{equation}
In Euclidean space, such approximations are well studied and can be calculated e.g. by multidimensional scaling (MDS), see also \cite{borg2005modern}. A survey of the hyperbolic setting, discussing in particular the properties of hyperbolic distance matrices, is given in \cite{tabaghi2020hyperbolic}.\\
An important special case is the \emph{network embedding problem}, where a (unweighted, undirected) graph $G = (V,E)$ is given. In this setting, the objects $\oo_1, \dotsc, \oo_n$ are given by the vertices $v_1, \dotsc, v_n$ of $G$ and as dissimilarities $d_{ij}$ between two vertices $v_i$ and $v_j$ for $i,j \in V$ we consider the length of the shortest path between $i$ and $j$, i.e., $D = [d_{ij}]$ is the graph distance matrix of $G$. We assume that the graph is connected, i.e., that the shortest path distances $d_{ij}$ are finite for any pair of vertices $v_i$ and $v_j$. \\
While the volume of Euclidean space expands polynomially, exponential expansion can be observed in hyperbolic space, see e.g. \cite{friedrich2019graph}. Because of this property, hyperbolic space is expected to give a better representation for hierarchical or tree-like structures than Euclidean geometry, and therefore is a favorable target space for graph embeddings, see e.g \cite{friedrich2019graph, kleinberg2007geographic}.

\subsection{Landmark-based network embeddings}
For large graphs or data sets, embedding methods encounter several limitations: First, the embedding method may become computationally too costly; second, the dissimilarity matrix $D$ may become too large to hold in memory; and third, the pre-computation of the dissimilarities $d_{ij}$ themselves may become infeasible. In the network embedding problem, for instance, computation of the full distance matrix of an unweighted, undirected graph requires $\cO(|V|^3)$ computations with the Floyd-Warshall algorithm (\cite{floyd1962algorithm}).\footnote{Other methods may be more efficient under additional assumptions on the graph structure, e.g., under bounds on the number of edges.} To alleviate this problem, \emph{landmark-based embedding methods} can be applied; see \cite{de2004sparse} for an application to MDS. For these methods, a comparatively small number of nodes, indexed by $L \subset V$ are designated as \emph{landmarks}. These landmarks are embedded into the target space, based on their pairwise distances. The remaining non-landmark points (indexed by $N = V \setminus L$) are then embedded by `triangulation', i.e., based on their distances with respect to the landmarks, but not to each other. Overall, only all shortest-paths sourced at landmark nodes have to be computed, reducing the cost of distance-computations to $\cO(|L| \cdot |V|^2)$. This represents an asymptotic increase in efficiency, \emph{if $|L|$ scales sublinear with $|V|$}, i.e. if the number of landmarks does not need to be scaled up proportional to $|V|$. In Theorem~\ref{thm:exact} we show that the number of landmarks required to recover point configurations in hyperbolic space $\cH_d$ depends only on the dimension $d$, not on the number of points. In this way, we provide a theoretical basis for the application of landmark methods to hyperbolic embeddings.
\subsection{Stress- vs. strain-based embeddings}
Stress-based embeddings solve the embedding problem \eqref{eq:embed} by minimizing the squared \emph{stress functional}
\begin{equation}\label{eq:stress_func}
\stress(\xx_1, \dotsc, \xx_n)^2 = \sum_{i,j=1}^n (d_{ij} - \dist^\kappa_H(\xx_i, \xx_j))^2.
\end{equation}
Equivalent loss functions are the root-mean-square error 
\begin{equation}\label{eq:RMSE}
RMSE = \sqrt{\frac{1}{n(n-1)} \stress(\xx_1, \dotsc, \xx_n)^2}
\end{equation}
and the relative embedding error
\begin{equation}\label{eq:REE}
REE = \stress(\xx_1, \dotsc, \xx_n) / \sqrt{\sum_{i,j=1}^n d_{ij}},
\end{equation}
which can be compared across different dissimilarity matrices $D$ and input sizes $n$.
Due to the properties of the hyperbolic distance $\dist^\kappa_H(\xx_i, \xx_j)$, minimizing \eqref{eq:stress_func} is a challenging non-convex optimization problem, which has to be solved by numerical minimization with no guarantee of convergence to the global minimum. Approaches by gradient descent and neural-network-based minimization are given in \cite{walter2004h,chamberlain2017neural}. Stress-minimization has been combined with a landmark-based approach and with advanced optimization methods by \cite{zhao2011fast} and \cite{chowdhary2017improved} under the names of \texttt{Rigel} and \texttt{HyPy} respectively. While \texttt{Rigel} uses the derivative-free Nelder–Mead simplex optimization method, \texttt{HyPy} provides an improvement by replacing the simplex optimization with iterative quasi-Newton minimization, for which efficient routines such as LBFGS \cite{zhu1997algorithm} can be used and supplied with the analytic gradient of stress, given in \cite[Eqs.~(3.1),(3.2)]{chowdhary2017improved}.
As an alternative, in our previous work \cite{keller2019hydra},  we have introduced \texttt{hydra}, a strain-based hyperbolic embedding method, based on minimization of the squared \emph{strain functional}
\begin{equation}\label{eq:strain}
\strain(\xx_1, \dotsc, \xx_n)^2 := \sum_{i,j} (\cosh(\sqrt{\kappa}\,d_{ij}) - (\xx_i \circ \xx_j))^2.
\end{equation}
The strain functional is obtained from the stress functional by the transformation of all distances by hyperbolic cosine. As shown in \cite{keller2019hydra},
$\texttt{hydra}$ has two important theoretical properties. The algorithm recovers any configuration of points in $d$-dimensional hyperbolic space up to isometry and it is guaranteed to return the globally optimal solutions of the strain-minimization problem \eqref{eq:strain}. In numerical experiments, \texttt{hydra} is faster than \texttt{HyPy} and \texttt{Rigel} by several orders of magnitudes, see~\cite{keller2019hydra}. While these are attractive properties, the strain functional \eqref{eq:strain} is not as interpretable as the stress functional \eqref{eq:stress_func}, which aligns with the popular `least-squares' paradigm of regression analysis and numerical approximation. Therefore, if we accept RMSE/REE as the target measure of embedding quality, there are two practical uses of the \texttt{hydra} method, cf. \cite{keller2019hydra}:
\begin{itemize}
\item Use \texttt{hydra} as a stand-alone method. This gives a very fast embedding method (several orders of magnitude faster than \texttt{HyPy/Rigel}), but with a relative embedding error (REE) that is approx. 1.0 to 1.5 times larger than the REE of stress-based embeddings produced by \texttt{HyPy/Rigel}.
\item Use the result of \texttt{hydra} as an initial condition for stress minimization. This gives an embedding method (called \texttt{hydra+} in \cite{keller2019hydra}) that is typically 20\% - 50\% faster than \texttt{HyPy/Rigel} and results in a smaller REE in all numerical tests with improvements up to 40\%.
\end{itemize}

We emphasize that these properties are limited by the fact that \texttt{hydra} is not a landmark-based method, and therefore can not be scaled up to the very large network instances considered for \texttt{HyPy} in \cite{chowdhary2017improved}.

\subsection{Strain-minimization with landmarks}
Both the \texttt{Rigel} algorithm of \cite{zhao2011fast}  and the \texttt{HyPy} algorithm of \cite{chowdhary2017improved} combine the method of stress-minimization with the landmark heuristic. That is, the stress functional is first minimized on landmarks only, and then the stress between landmark and non-landmark-nodes is minimized. In addition, \texttt{HyPy} applies heuristics of gradually expanding the landmark set during optimization and of restarting the optimization procedure from multiple random initial point configurations. Due to the landmark approach, \texttt{Rigel/HyPy} are able to deal with very large instances of the network embedding problem with millions of nodes. The \texttt{hydra} method, as proposed in \cite{keller2019hydra}, on the other side, requires knowledge of the full dissimilarity matrix $D = [d_{ij}]$ and is therefore only applicable to small or medium sized networks. This is our motivation, to introduce here a landmarked version of \texttt{hydra}, and to show that strain-minimization can be combined with the landmark heuristic while retaining all its attractive theoretical properties. In particular, we show in Theorem~\ref{thm:exact} that the necessary size of the landmark set only depends on the intrinsic dimension $d$ of the embedding space, but not on the total number of nodes (or other objects) to be embedded. This gives, for the first time, a sound theoretical basis for the application of the landmark heuristic for hyperbolic embedding problems. 
In terms of practical applications, we demonstrate in Section~\ref{sec:numerics} that the landmarked hydra method \texttt{L-hydra} can be scaled up to all network embedding problems considered in \cite{chowdhary2017improved} with up to almost 4 million nodes. The advantages of \texttt{hydra} as compared with \text{HyPy/Rigel} are effectively retained in the landmarked setting (see Section~\ref{sec:numerics} for details): 
\begin{itemize}
\item As a stand-alone method, \texttt{L-hydra} is a very fast embedding method (about ten times faster than \texttt{HyPy} with landmarks), but with a higher relative embedding error than the stress-based embedding produced by \texttt{HyPy}. 
\item Using the result of \texttt{L-hydra} as an initial condition for stress minimization (we call this method \texttt{L-hydra+}) gives an embedding method that is both faster than \texttt{HyPy} and produces embeddings with smaller error.
\end{itemize}

\section{The landmarked hydra algorithm}
\subsection{Formulation of the embedding problem}

Consider a set of objects $\oo_1, \dotsc, \oo_n$, which are partitioned into landmarks $(\oo_i)_{i \in L}$ and non-landmarks $(\oo_i)_{i \in N}$. The pairwise dissimilarity $d_{ij}$ of $\oo_i$ and $\oo_j$ is only assumed to be known when $\oo_i$ or $\oo_j$ is a landmark. By permuting objects, the dissimilarity matrix $D$ can be arranged as $D = \begin{psmallmatrix} D_{L} & D_N^\top \\ D_{N} & D_R \end{psmallmatrix}$ where only $D_L$  -- containing all pairwise dissimilarities between landmark nodes -- and $D_N$ -- containing all dissimilarities between one landmark and one non-landmark node -- are known.\\
We first discuss the problem of exact recovery in hyperbolic space. That is, we assume that for a given embedding dimension $d$, the objects $\oo_i$ can be perfectly described by points $\xx_i$ in the hyperbolic manifold $\cH_d$, i.e. the dissimilarity $d_{ij}$ between any pair of objects $(\oo_i, \oo_j)$ can be described exactly by the hyperbolic distance between the corresponding points $\xx_i,\,\xx_j \in \cH_d$:
\[ d_{ij} = \dist^\kappa_H(\xx_i, \xx_j)  = \frac{1}{\sqrt{\kappa}} \arcosh (\xx_i \circ \xx_j). \]
Transforming all distances by the hyperbolic cosine, we obtain
\begin{equation}\label{eq:transf_dist}
\cosh(\sqrt{\kappa}\,d_{ij}) = \xx_i \circ \xx_j.
\end{equation}
We set $A = \cosh \left( \sqrt{\kappa}\,D\right)$, where the hyperbolic cosine is applied elementwise to the dissimilarity matrix $D$. Setting $l = |L|$ and $m = |N|$ we write
\[X = \left( \begin{array}{r} X_{L} \\ X_{N} \end{array}\right) = \left(\xx_1, \dotsc, \xx_l, \xx_{l+1}, \dotsc, \xx_{l+m} \right)^\top \in \RR^{(l+m) \times (d+1)}\]
for the coordinate matrix of some points $\xx_1, \dotsc, \xx_{l+m}$ in $\cH_d$. Finally, let $J$ be the $(d+1) \times (d+1)$ diagonal matrix
\begin{equation*}
J = \diag(1,-1, \dotsc, -1),
\end{equation*}
cf. \cite[\S3.1]{ratcliffe2006foundations}. Equation \eqref{eq:transf_dist} for $i \in L \cup N$ and $j \in L$ can now be written in compact form as
\begin{equation}\label{eq:opt_compact}
 \left( \begin{array}{r} A_{L} \\ A_{N} \end{array}\right) = \left( \begin{array}{r} X_{L} \\ X_{N} \end{array}\right) J\,X_L^\top.
\end{equation}
Thus, the coordinate matrix $X = \begin{psmallmatrix} X_{L} \\ X_{N} \end{psmallmatrix}$ can be recovered from the known dissimilarities $D_L, D_N$ if we can solve \eqref{eq:opt_compact} for $X_L$ and $X_N$. Due to its block structure, equation \eqref{eq:opt_compact} can be split into two parts and solved in two stages, where the coordinates for landmarks must be determined from the first equation and the coordinates of non-landmarks from the second one:
\begin{subnumcases}{\label{eq:stagewise}}
  A_{L} = X_{L} J X_{L}^\top \label{eq:min_L} \label{eq:stagewise_L}\\
  A_{N} = X_{N} J X_{L}^\top \label{eq:min_N}. \label{eq:stagewise_N}.
\end{subnumcases}
However, we can only hope to recover $(X_L,X_N)$ from these equations up to hyperbolic isometry,  i.e., up to a distance-preserving transformation $\psi: \cH_d \to \cH_d$. From \cite[Sec.~3.1,3.2]{ratcliffe2006foundations}, any hyperbolic isometry can be written as
\begin{equation}\label{eq:isometry}
\psi(\xx) = T\xx,
\end{equation}
where $T \in \RR^{(d +1) \times (d+1)}$ is an invertible matrix, which has $T_{11} > 0$ and satisfies
\begin{equation}\label{eq:Lorentz_iso}
T^\top J T = TJT^\top = J.
\end{equation}
Matrices with this property are called \emph{positive Lorentz matrices} and both $\cH_d$ and $\Lor_+$ are invariant under such transformations. Our observation can now be formalized as follows:
\begin{lem}Let $T \in \RR^{(d +1) \times (d+1)}$ be a positive Lorentz matrix and $(X_L, X_N)$  a solution of \eqref{eq:stagewise}. Then $(\tilde X_L, \tilde X_N) = (X_L T, X_N T)$ is also a solution of \eqref{eq:stagewise}. If $X_L$ and $X_N$ are hyperbolic coordinate matrices of points $(\xx_1, \dotsc, \xx_{l+m}) \in \cH_d$ , then $\tilde X_L$ and $\tilde X_N$ are also hyperbolic coordinate matrices of hyperbolically isometric points $(\tilde \xx_1, \dotsc, \tilde \xx_{l+m}) \in \cH_d$.
\end{lem}
\begin{proof}
The first statement follows directly from 
\[\tilde X_L  J \tilde X_L^\top  = X_L T J T^\top X_L^\top = X_L J X_L^\top  = A_L\]
and 
\[\tilde X_N  J \tilde X_L^\top  = X_N T J T^\top X_L^\top = X_N J X_L^\top  = A_N.\]
The second statement follows from the invariance of $\cH_d$ under the positive Lorentz matrix $T$. 
\end{proof}

\subsection{Solving the embedding problem}
In general, one cannot expect given dissimilarities $d_{ij}$ to be represented exactly by hyperbolic distances in $\cH_d$, and therefore should not expect an exact solution to \eqref{eq:stagewise}. For this reason, we replace \eqref{eq:stagewise} by its relaxation
\begin{subnumcases}{}\label{eq:relaxation}
\hat X_L = \argmin \set{\norm{A_L - X_L J X_L^\top}_F: X_L \in \RR^{l\times (d+1)}} \label{eq:relaxation_L}\\
\hat X_N = \argmin \set{\norm{A_N - X_N J X_L^\top}_F: X_N \in \RR^{m\times (d+1)}},\label{eq:relaxation_N}
\end{subnumcases}
where $\norm{.}$ indicated the Frobenius norm on matrices. Note that the relaxation consists of two parts: First, equality has been replaced by minimality of the Frobenius distance. Second, $\hat X_L$ and $\hat X_N$ are no longer constrained to coordinate matrices of points in $\cH_d$, but of points in the ambient Lorentz space $\Lor$.

The \texttt{L-hydra} algorithm consists of the stepwise solution of \eqref{eq:relaxation}:
\begin{itemize}
\item In step 1 and 2, the minimization problem \eqref{eq:relaxation_L} is solved by a matrix Eigendecomposition. This is equivalent to the \texttt{hydra} embedding of \cite{keller2019hydra}, applied to the landmarks only.
\item In step 3, the least-squares problem \eqref{eq:relaxation_N} is solved using the Moore-Penrose inverse of $\hat X_L J$. 
\end{itemize}
The details of each step are listed in Algorithm~\ref{algo:hydra_lm}.

\begin{algorithm}
\caption{\texttt{L-hydra}$\,(D_L, D_N,d,\kappa)$}
\label{algo:hydra_lm}
\begin{description}
\item[\textbf{Input}]
\begin{itemize}
\item Two matrices $D_L \in \Rplus^{l \times l}$ and $D_N \in \Rplus^{m \times l}$, of which the first represents the dissimilarities between pairs of landmarks and the second between pairs of landmarks and non-landmarks. The matrix $D_L$ must be symmetric with zero diagonal.
\item Embedding dimension $d \le (l+m-1)$
\item Parameter $\kappa > 0$; the negative of the hyperbolic curvature $-\kappa$. 
\end{itemize}
\item[\textbf{Step 1}] Set\vspace{-8pt}
\begin{subnumcases}{\label{eq:Acosh}}
A_L  =  \cosh\left(\sqrt{\kappa} D_L\right)\\
A_N  =  \cosh\left(\sqrt{\kappa} D_N\right)
\end{subnumcases}
with $\cosh$ applied elementwise,  and compute the Eigendecomposition of the matrix $A_L$
\begin{equation}\label{eq:Eigen}
A_L = Q \Lambda_L Q^\top,
\end{equation}
where $\Lambda_L$ is the diagonal matrix of the Eigenvalues $\lambda_1 \ge \dotsm \ge \lambda_l$ and the columns of $Q$ are the cor\-responding Eigenvectors $\qq_1, \dotsc, \qq_l$.
\item[\textbf{Step 2}] Assuming that the last $d$ Eigenvalues are negative, allocate the $l \times (d+1)$-matrix
\begin{equation}\label{eq:X_L}
\hat X_L := \left[\sqrt{\lambda_1}\,\qq_1 \quad \sqrt{(-\lambda_{l-d+1})}\,\qq_{l-d+1} \quad \dotsm \quad \sqrt{(-\lambda_{l})}\,\qq_{l} \right].
\end{equation}
If not all of the last $d$ Eigenvalues are negative, return \textbf{Null}. The algorithm must be rerun with smaller embedding dimension $d$ for non-null result.
\item[\textbf{Step 3}] Allocate the $m \times (d+1)$-matrix
\begin{equation}\label{eq:X_N}
\begin{aligned}
\hat X_N :&= A_N \hat X_L \diag\left(\frac{1}{\lambda_1}\, \quad -\frac{1}{\lambda_{l-d+1}}\, \quad \dotsm \quad -\frac{1}{\lambda_l}\, \right) \\
    &= A_N \left[\frac{\qq_1}{\sqrt{\lambda_1}}\, \quad -\frac{\qq_{l-d+1}}{\sqrt{-\lambda_{l-d+1}}}\, \quad \dotsm \quad -\frac{\qq_l}{\sqrt{-\lambda_l}}\, \right].
\end{aligned}
\end{equation}
\item[\textbf{Return}] Matrix $\hat{X} = \begin{pmatrix} \hat{X}_L \\ \hat{X}_N \end{pmatrix}$ whose rows are coordinates in positive Lorentz space $\Lor_+$.
\end{description}
\end{algorithm}
\pagebreak
\subsection{Theoretical Properties}
The key theoretical properties of the \texttt{L-hydra} algorithm are summarized in the following theorems.

\begin{thm}[Exact Recovery] \label{thm:exact}
Let $\xx_1, \dotsc, \xx_{m+l}$ be points in hyperbolic $d$-space $\cH_d$, of which the first $l \geq d$ are designated as landmarks. Assume that the landmarks are not all contained in a single hyperplane of $\cH_d$ and let $D = [d_{ij}] = [\dist^\kappa_H(\xx_i,\xx_j)]$ be the matrix of their hyperbolic distances with curvature $-\kappa$, partitioned as 
\[ D = \begin{pmatrix} D_{L} & D_N^\top \\ D_{N} & D_R \end{pmatrix}.\]
Then the algorithm $\texttt{L-hydra}\,(D_L,D_N,d,\kappa)$ recovers all points $\xx_1, \dotsc, \xx_{l+m}$ up to isometry. That is, the rows $\hat \xx_1, \dotsc, \hat \xx_{l+m}$ of the matrix $\hat X$ returned by $\texttt{L-hydra}\,(D_L,D_N,d,\kappa)$ are points in $\cH_d$ and satisfy
\begin{equation}\label{eq:exact}
    \dist_H(\hat \xx_i, \hat \xx_j) = \dist_H(\xx_i, \xx_j) \qquad \text{for all } i,j = 1, \dotsc, l+m.
 \end{equation}
\end{thm}

\begin{thm}[Optimal Approximation] \label{thm:optimal}
Suppose that $A_L = \cosh\left(\sqrt{\kappa} D_L\right)$ has at least $d \in \NN$ strictly negative Eigenvalues. Then, the matrices $\hat X_L, \hat X_N$ returned by the algorithm \texttt{L-hydra}$\,(D_L,D_N,d,\kappa)$ solve the minimization problem \eqref{eq:relaxation}. Moreover, the first columns of $\hat X_L$ and $\hat X_N$ are strictly positive; equivalently, all rows of $\hat X_L$ and $\hat X_N$ represent points in positive Lorentz space $\Lor_+$.
\end{thm}
The first part of the result, i.e., the optimality of the embedding of landmarks, follows from \cite{keller2019hydra}. For convenience, we provide a self-contained proof of both parts:
\begin{proof}
Let $A_L = Q \Lambda_L Q^\top$ be the Eigendecomposition of $A_L = X_L J X_L^\top$, where $\Lambda$ is the diagonal matrix of the Eigenvalues $\lambda_1 \ge \lambda_2 \ge \dotsm \ge \lambda_l$ of $A_L$ in descending order and $Q$ is unitary, i.e., $Q^\top Q = I$. By the unitary invariance of the Frobenius norm, solving the minimization problem \eqref{eq:relaxation_L} is equivalent to solving
\begin{equation}\label{eq:minimize_diag}
\hat Y_L = \argmin \set{\norm{Y_L J Y_L^\top - \Lambda}_F: Y_L \in \RR^{l \times (d+1)}},
\end{equation}
where $\hat Y_L$ is related to $\hat X_L$ by $X_L = Q \hat Y_L$. We claim that \eqref{eq:minimize_diag} is minimized by 
\begin{equation}\label{eq:X_L_tilde}
\hat Y_L := \left[\sqrt{\lambda_1}\,\ee_1 \quad \sqrt{(-\lambda_{l-d+1})}\,\ee_{l-d+1} \quad \dotsm \quad \sqrt{(-\lambda_{l})}\,\ee_{l} \right].
\end{equation}
To see this, note that on the one hand, 
\[\norm{\hat Y_L J \hat Y_L^\top - \Lambda} = \sum_{i=2}^{l-d} \lambda_i^2. \]
On the other hand, for an arbitrary $Y_L \in \RR^{l \times (d+1)}$, denote the Eigenvalues of $Y_L J Y_L^\top$ by $\eta_1 \ge \eta_2 \ge \dotsm \ge \eta_l$. By Sylvester's law of inertia, owing to the structure of $J$, only $\eta_1$ is strictly positive, only the last $d$ Eigenvalues $\eta_{l-d+1}, \dotsc, \eta_l$ are strictly negative, and all other Eigenvalues must be zero. Hence, using the Wielandt-Hofmann inequality, we obtain
\begin{equation*}
\norm{Y_L J Y_L^\top - \Lambda}_F \ge \sum_{i=1}^l \left(\eta_i - \lambda_i \right)^2 \ge \sum_{i=2}^{l-d} \lambda_i^2 = \norm{\hat Y_L J \hat Y_L^\top - \Lambda}_F
\end{equation*} 
for all $Y_L \in \RR^{l \times (d+1)}$, showing the optimality of $\hat Y_L$ for \eqref{eq:minimize_diag}. Transforming back to $\hat X_L = Q \hat Y_L$ yields \eqref{eq:X_L} and shows that $\hat X_L$ solves \eqref{eq:relaxation_L}. It remains to solve \eqref{eq:relaxation_N}, which is a least-squares problem, whose solution can be written as
\begin{equation} \label{eq:Pseudo}
\hat X_N = A_N (J \hat X_L^\top)^+,
\end{equation}
where $(J \hat X_L^\top)^+$ denotes the Moore-Penrose-Pseudoinverse; see \cite[Theorem.~(8.1) and Remark.~(8.2)]{laub2005matrix}. Since $J \hat X_L^\top$ has full row rank, we can write 
\begin{equation}
(J \hat X_L^\top)^+ = \hat X_L J(J \hat X_L^\top \hat X_L J)^{-1} = \hat X_L(\hat X_L^\top \hat X_L)^{-1} J = \hat X_L \diag\left(\frac{1}{\lambda_1}\, \quad -\frac{1}{\lambda_{l-d+1}}\, \quad \dotsm \quad -\frac{1}{\lambda_l}\, \right).
\end{equation}
Inserting this into \eqref{eq:Pseudo} yields \eqref{eq:X_N}, completing the main part of the proof. Finally, note that by the Perron-Frobenius theorem, the leading Eigenvector $\qq_1$ of the positive matrix $A_L$ must also be positive. Together with \eqref{eq:X_N}, this shows that the first columns of both $\hat X_L$ and $\hat X_N$ are positive. 
\end{proof}

\begin{proof}[Proof of Thm.~\ref{thm:exact}]
Denote by $X = \begin{psmallmatrix}X_L\\X_N\end{psmallmatrix}$ the coordinate matrix of the original points $\xx_1, \dotsc, \xx_{l+m}$ and by $\hat X = \begin{psmallmatrix}\hat X_L\\\hat X_N\end{psmallmatrix}$ the coordinate matrix of the points $\hat \xx_1, \dotsc, \hat \xx_{l+m}$ returned by \texttt{L-hydra}. Since the points $\xx_1, \dotsc, \xx_l$ are not all contained in a single hyperplane of $\cH_d$, it follows from \cite[\S3.2]{ratcliffe2006foundations} that there are $d+1$ among them which are linearly independent as points in $\RR^{d+1}$, and thus that $X_L$ has the full column rank $d+1$. Let $\lambda_1 \ge \lambda_2 \ge \dotsm \ge \lambda_{l}$ be the Eigenvalues of $A_L = X_L J X_L^\top$. By the same argument as in the proof of Theorem~\ref{thm:optimal} it follows that only $\lambda_1$ is strictly positive, only the $d$ last Eigenvalues $\lambda_{l-d+1}, \dotsc, \lambda_l$ are strictly negative and all others are zero. From the proof of Theorem~\ref{thm:optimal} we find that the residual of \eqref{eq:relaxation_L} is
\[\norm{A_L - \hat X_L J \hat X_L^\top}^2_F = \sum_{i=2}^{l-d} \lambda_i^2 = 0,\]
i.e., the embedding of the landmarks is exact. This means that $X_L J X_L^\top = \hat X_L J \hat X_L^\top$, i.e.,  $\hat X_L$ and $X_L$ are hyperbolically isometric and by \eqref{eq:isometry}, there exists a positive Lorentz matrix $T$ such that
\[\hat X_L = X_L T.\]
For the Moore-Penrose-Pseudoinverse $(J X_L^\top)^+$ it follows that 
\[(J \hat X_L^\top)^+ = (J X_L^\top)^+ T.\]
Moreover, since $X_L$ has the full column rank, $(J X_L^\top)^+$ is a right-inverse of $(J X_L^\top)$, and we have
\begin{equation*} 
\hat X_N = A_N (J \hat X_L^\top)^+ = A_N (J X_L^\top)^+ T = X_N J X_L^\top (J X_L^\top)^+ T = X_N T.
\end{equation*}
Together, it follows that $\hat X = X T$, i.e., $\hat X$ and $X$ are hyperbolically isometric and \eqref{eq:exact} must hold for all (both landmark and non-landmark) points.
\end{proof}

\begin{cor}[Consistency of landmark and non-landmark embedding] \label{thm:consistency}
The landmark embedding is consistent with the non-landmark embedding, meaning that re-embedding the $j$-th landmark point ($j \in \set{1, \dotsc, l}$) as a non-landmark will not change its representation $\hat \xx_j$, or equivalently 
\begin{equation}\label{eq:consistent}
\hat X_L  = A_L \left[\frac{\qq_1}{\sqrt{\lambda_1}}\, \quad -\frac{\qq_{l-d+1}}{\sqrt{-\lambda_{l-d+1}}}\, \quad \dotsm \quad -\frac{\qq_l}{\sqrt{-\lambda_l}}\, \right].
\end{equation}
\end{cor}
\begin{rem}
Re-embedding of landmarks is equivalent to inserting the transformed landmark dissimilarity matrix $A_L$, instead of $A_N$, into \eqref{eq:X_N}. Hence, consistency can be expressed by equation \eqref{eq:consistent}.
\end{rem}
\begin{proof}
For each eigenvector $\qq_j$ with $j=l-d, \dotsc, l$ of $A_L$ we have
\[ \frac{A_L \qq_j}{\sqrt{-\lambda_j}} = \frac{\lambda_j \qq_j}{\sqrt{-\lambda_j}} = \sqrt{-\lambda_j} \qq_j, \]
and -- omitting the minus sign -- also for $j=1$. Applying this column-by-column and taking into account \eqref{eq:X_L} we obtain \eqref{eq:consistent}.
\end{proof}

\subsection{Practical implementation}\label{sec:practical}
After having analyzed the theoretical properties of \texttt{L-hydra}, we point out some more practical issues in the implementation of \texttt{L-hydra}:

\begin{description}
\item[\textbf{Reduced Eigendecomposition.}] In \eqref{eq:X_L} only the first and the last $d$ Eigenvalues and Eigenvectors of the matrix $A_L$ are needed. There are efficient numerical routines to perform such a reduced Eigendecomposition without computing the full Eigendecomposition of $A_L$. 
\item[\textbf{Optimizing Curvature.}] The \texttt{L-hydra} algorithm treats the curvature parameter $\kappa$ as fixed input. In a practical implementation it is usually desirable to also optimize over $\kappa$ by running \texttt{L-hydra} for several different values of $\kappa$. In a similar way, curvature is optimized in \texttt{HyPy} by treating it as an additional free variable of the stress functional \eqref{eq:stress_func}.
\item[\textbf{Projection to $\cH_d$.}] \texttt{L-hydra} returns points $\hat \xx_1, \dotsc, \hat \xx_{l+m}$ in positive Lorentz space $\Lor$, which are `reasonably close' to $\cH_d$, due to Theorem~\ref{thm:exact}, but typically not exactly located on the hyperboloid $\cH_d$. To obtain points on $\cH_d$, a projection method must be applied, cf. \cite{keller2019hydra}. Here, we project parallel to the $x_1$-axis, i.e. given $\xx \in 
\Lor$ we set 
\[\tilde x_1 := \sqrt{1 + x_2^2 + \dotsm + x_{d+1}^2}\]
to obtain a projected point $\tilde \xx = (\tilde x_1, x_2, \dots, x_{d+1}) \in \cH_d$ for each point returned by \texttt{L-hydra}.
\item[\textbf{Analysis of Runtime.}] Recall that the total number of points, the number of landmarks and the number of non-landmarks are denoted by $n$, $l$ and $m = n - l$ respectively. For Step 1 and 2 of hydra, we expect a runtime of $\cO(l^\alpha)$ with $\alpha$ slightly above, but close, to 2, cf.~\cite{keller2019hydra, hogben2006handbook}, using reduced Eigendecomposition. For Step 3 of hydra, we expect a runtime of $\cO(lm)$ for the multiplication of a $m \times l$ and a $l \times (d+1)$ matrix, resulting in an overall runtime of $\cO(l^{\alpha-1} n)$. On the basis of Theorem~\ref{thm:exact}, the necessary size of the landmark set depends only on the intrinsic dimension $d$, but not on the total number of points to be embedded. This implies, that the runtime of \texttt{L-hydra} scales effectively \emph{linear} as $\cO(n)$ in the number of points to be embedded.
\item[\textbf{The} \texttt{L-hydra+} \textbf{method}.] If minimization of the stress functional \eqref{eq:stress_func} (or, equivalently, relative embedding error \eqref{eq:REE}) is the ultimate goal, then the result of \texttt{L-hydra} can be used as an initial condition for its numerical minimization. Effectively, this corresponds to a chaining of \texttt{L-hydra} and \texttt{HyPy}, where \texttt{L-hydra} replaces the random initial condition of \texttt{HyPy}.
\end{description}

\section{Numerical Results}\label{sec:numerics}
We test the practical performance of \texttt{L-hydra} and \texttt{L-hydra+} on five social networks from SNAP (\cite{yang2012defining}), with $\sim$300,000 up to almost 4 million nodes and $\sim$1 million up to $\sim$117 million edges; see Table 1 for a brief description of the networks. The same networks were used in \cite{chowdhary2017improved} to evaluate the performance of \texttt{HyPy}, which serves as a benchmark for our methods. All networks are unweighted and undirected. The $l = 100$ landmark nodes are chosen with probability proportional to the node degrees without replacement as proposed by \cite{chowdhary2017improved}, which combines the benefits of random sampling and the selection of highest-degree nodes. All shortest-path distances are computed using SNAP for Python (\cite{leskovec2016SNAP}). \texttt{L-hydra} and \texttt{L-hydra+} were implemented in Python as described in Algorithm~\ref{algo:hydra_lm} and Section~\ref{sec:practical}, i.e., with curvature optimization, reduced Eigendecomposition and projection to $\cH_d$, using parallelized routines for matrix operations. For \texttt{HyPy}, we used the Python code kindly supplied by the authors of \cite{chowdhary2017improved}, using the same parallelization as for \texttt{L-hydra} and \texttt{L-hydra+}. All calculations were performed on a Dual socket Intel server using an Intel Xeon CPU E5-2680 v3 with 12 cores at 2.50GHz and with 64GB RAM.\\

\begin{table}[htpb]
\begin{footnotesize}
\begin{tabular}{@{}lp{0.3\textwidth}lll@{}}
\toprule
Network&Description&Source&\#\,Nodes&\#\,Edges \\
 & & \url{http://snap.stanford.edu/data/...} & & \\
\midrule
Amazon & Network based on the 'Customers Who Bought This Item Also Bought' feature of the Amazon website & \url{.../com-Amazon.html} & $334,863$ & $925,872$ \\
DBLP & Collaboration network of the DBLP computer science bibliography & \url{.../com-DBLP.html} & $317,080$ & $1,049,866$ \\
YouTube & Friendships in the YouTube social network & \url{.../com-Youtube.html} & $1,134,890$ & $2,987,624$ \\
Live Journal & Friendship network of a free online blogging community & \url{.../com-LiveJournal.html} & $3,997,962$ & $34,681,189$ \\
Orkut & Friendships in an online social network & \url{.../com-Orkut.html} & $3,072,441$ & $117,185,083$ \\
\bottomrule
\end{tabular}
\vspace{0.5em}
\caption{Description of networks used for numerical experiments}\label{table:networks}
\end{footnotesize}
\end{table}

The embedding results for all methods are shown in Figures \ref{fig:mse1} and \ref{fig:mse2}. For each method and network, we calculate the network embedding for dimensions $d = 2, 3, \dots , 10$. Following \cite{chowdhary2017improved}, we show the relative embedding error \eqref{eq:REE} in three variations: calculated over all pairs of \emph{landmark} nodes; calculated over all pairs consisting of one landmark and one \emph{non-landmark} node; and over 100,000 randomly chosen pairs of non-landmark nodes (`\emph{validation error}'). Note that due to the size of the networks, calculation of the full embedding error over all node pairs is computationally infeasible, but the validation error should provide a good proxy. Since \texttt{HyPy} starts from a random initial condition, we show the $5\%-$ and $95\%-$quantiles in addition to its average result over 20 runs.\\
The runtimes for embedding into dimension $d=2$ are shown in Figure \ref{fig:times}; results for other dimensions are similar. Note that the runtime can be split into distance calculation -- which is the same for all methods -- and embedding calculation -- which differs from method to method. We summarize our observations as follows:

\begin{description}
\item[\textbf{Consistency with \cite{chowdhary2017improved}.}] The results obtained for \texttt{HyPy} are consistent with the results reported in \cite{chowdhary2017improved}; note that we report REE while  RMSE is reported in \cite{chowdhary2017improved}.
\item[\textbf{Performance of} \texttt{L-hydra}.] As expected, the strain-minimizing algorithm \texttt{L-hydra} typically achieves poorer results in terms of the relative embedding error (REE) than the stress-minimizing algorithms \texttt{L-hydra+} and \texttt{HyPy}.  Compared to the average performance of \texttt{HyPy}, the validation REE of \texttt{L-hydra} is on average (median over all networks and dimensions) $2.17$ times larger. However, the total runtime of \texttt{L-hydra} (the bulk of which is spent on distance calculation) is between 5 and 10 times faster than the average runtime of \texttt{HyPy}. In essence, replacing \texttt{HyPy} by \texttt{L-hydra} trades a moderate increase in embedding error against a substantial decrease in computation time.
\item[\textbf{Performance of} \texttt{L-hydra+}.] The embedding result of \texttt{L-hydra+} -- which combines \texttt{L-hydra} with stress minimization -- is consistently (over all networks, dimensions and error types) better than the average result of \texttt{HyPy}, and in the large majority of cases even better than the $5\%-$quantile of \texttt{HyPy} results. Averaged over all dimensions and networks, the validation REE of \texttt{L-hydra+} is $12\%$ smaller than the validation REE of \texttt{HyPy}. At the same time, \texttt{L-hydra+} is faster than \texttt{HyPy} for all network embeddings except for YouTube, with a relative difference in average runtime from $-52.5\%$ (DBLP network) to $+5.7\%$ (YouTube network). 
\item[\textbf{Stability and Reproducibility.}] Both $\texttt{L-hydra}$ and \texttt{L-hydra+} do not suffer from the noise introduced by the random initial condition of \texttt{HyPy}, which -- in particular for small dimension -- causes a notable variation in embedding quality of \texttt{HyPy} and yields different embedding results upon each rerun.
\item[\textbf{Runtime analysis.}] In panels A, B of Figure~\ref{fig:times} we show regression lines to estimate the exponent $\alpha$ in the conjectured complexity $\mathcal{O}(n^\alpha)$ for distance calculation and embedding. For distance calculation we obtain the estimate $\alpha_\text{dist} \approx 1.31$, showing that due to the sparsity of the networks, shortest path computations are more efficient than the upper bound $\alpha = 2$ obtained from the Floyd-Warshall method. For the embedding step of \texttt{L-hydra+} we estimate $\alpha_\text{L-hydra+} \approx 1.32$, while for \texttt{L-hydra} we obtain $\alpha_\text{L-hydra+} \approx 0.99$, confirming the linear scaling of the method.
\end{description}

\begin{figure}
\centering
\includegraphics[width=0.99\textwidth]{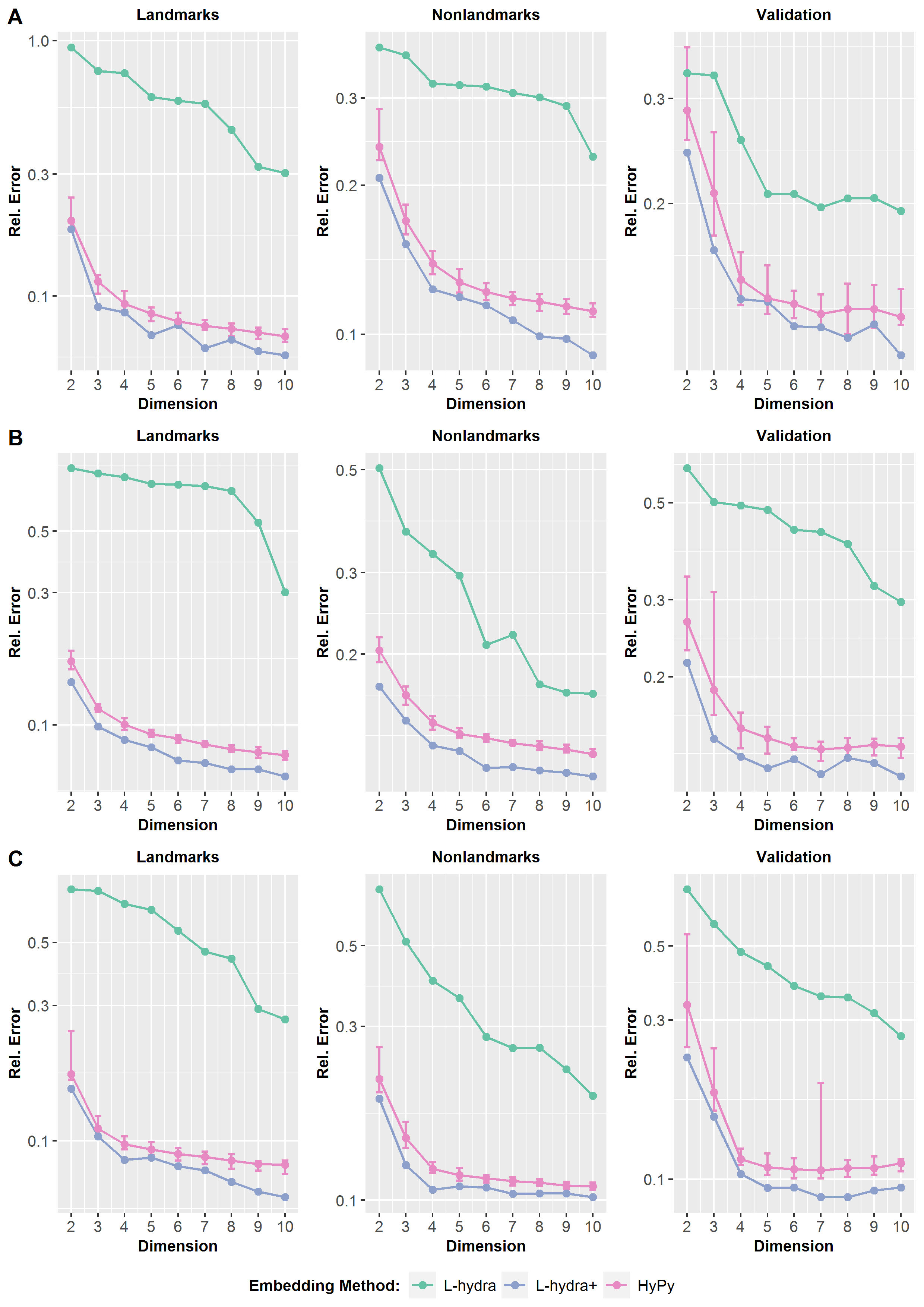}
\caption{Comparison of the embedding performance of \texttt{HyPy}, \texttt{L-hydra} and \texttt{L-hydra+} on the real network data sets (A) Amazon, (B) DBLP and (C) YouTube. Embedding quality is measured by the relative embedding error \eqref{eq:REE} over pairs of landmark nodes (left column), pairs of landmark-non-landmark nodes (middle column) and over 100,000 randomly selected validation pairs of non-landmark nodes (right column). The same $l=100$ landmarks are used for all three methods. For \texttt{L-hydra} a $5 - 95\%$ error bar is shown, corresponding to 20 runs with randomized initial condition.}
\label{fig:mse1}
\end{figure}

\begin{figure}
\centering
\includegraphics[width=0.99\textwidth]{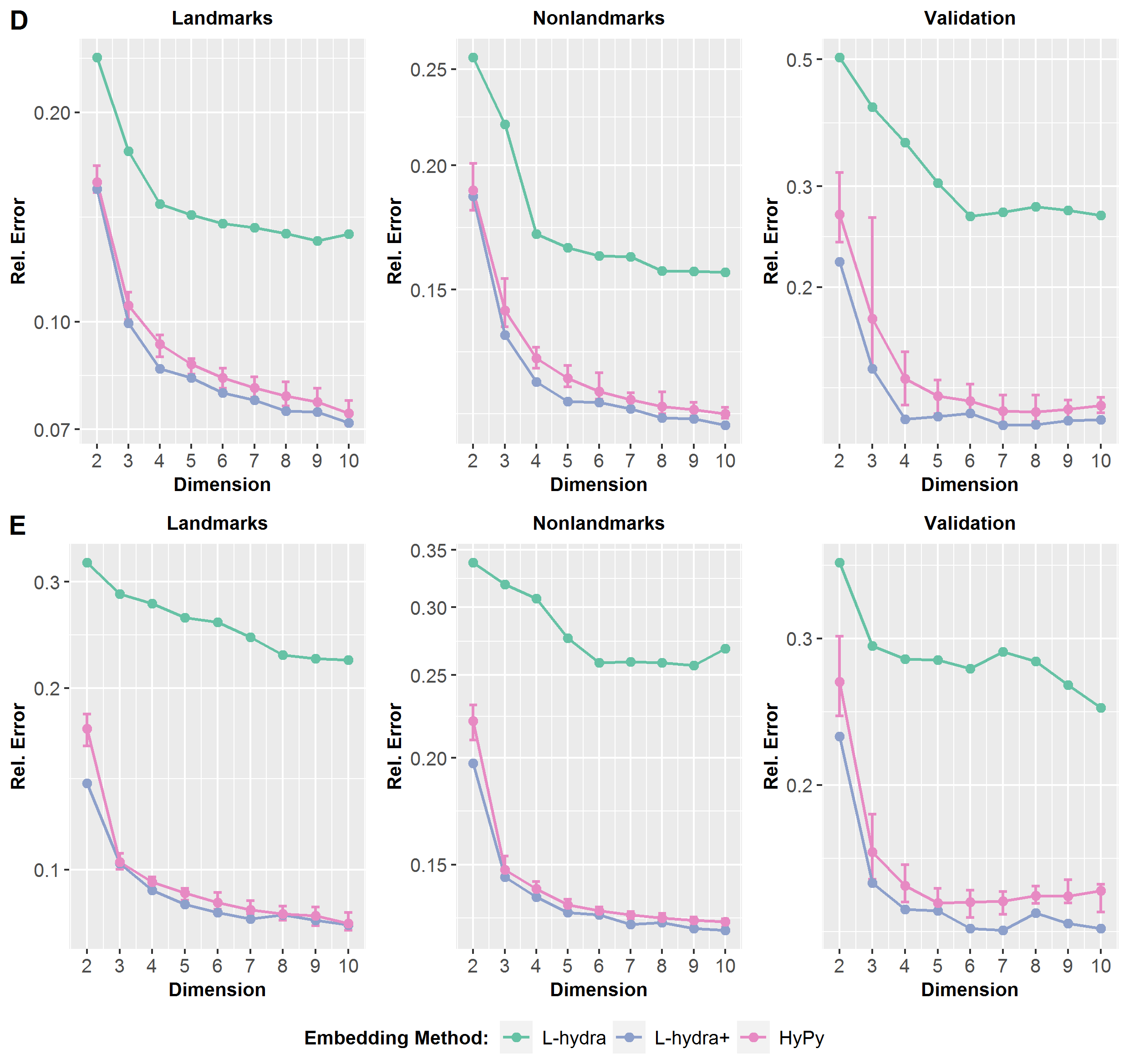}
\caption{Comparison of the embedding performance of \texttt{HyPy}, \texttt{L-hydra} and \texttt{L-hydra+} on the real network data sets (D) LiveJournal and (E) Orkut. Embedding quality is measured by the relative embedding error \eqref{eq:REE} over pairs of landmark nodes (left column), pairs of landmark-non-landmark nodes (middle column) and over 100,000 randomly selected validation pairs of non-landmark nodes (right column). The same $l=100$ landmarks are used for all three methods. For \texttt{L-hydra} a $5 - 95\%$ error bar is shown, corresponding to 20 runs with randomized initial condition.}
\label{fig:mse2}
\end{figure}

\begin{figure}
\centering
\includegraphics[width=0.99\textwidth]{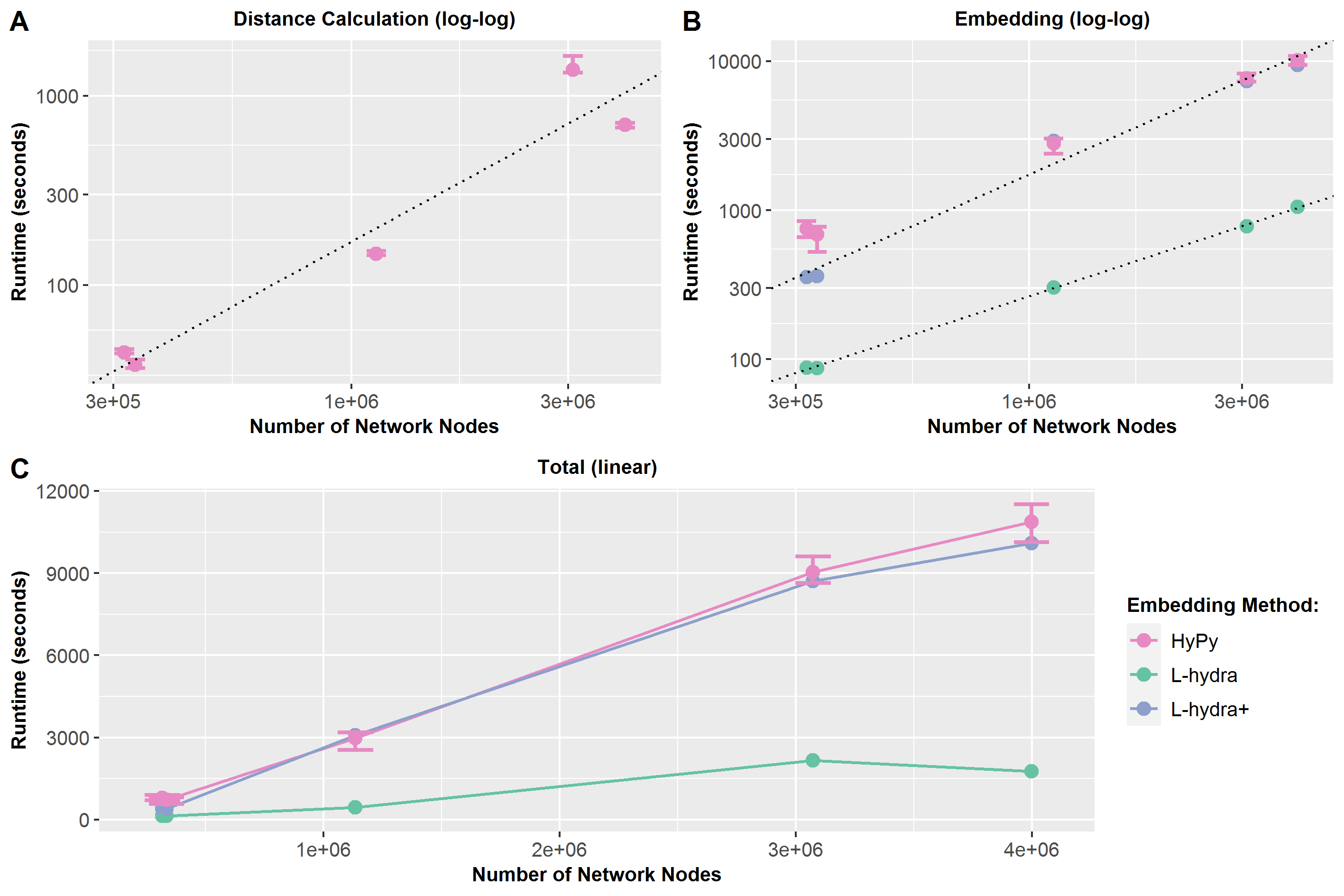}
\caption{Computation time (in seconds) of the hyperbolic embedding methods \texttt{HyPy}, \texttt{L-hydra} and \texttt{L-hydra+} applied to the five networks listed in \ref{table:networks}. For \texttt{HyPy}, average computation time and a $5 - 95\%$ error bar is shown, corresponding to 20 runs with randomized initial condition. Regression lines for \texttt{L-hydra} and \texttt{L-hydra+} are added in the doubly logarithmic plots (A) and (B).}
\label{fig:times}
\end{figure}


\bibliographystyle{plain}
\bibliography{references}
\end{document}